\documentclass{article}
	
	\usepackage{amsmath,amssymb,amsfonts,amsthm}
	\usepackage{graphicx}
	\usepackage{cite}
	\usepackage{fullpage}
	\usepackage{commath}

	\newtheorem{theorem}{Theorem}[section]
	\newtheorem{lemma}[theorem]{Lemma}

	\theoremstyle{definition}
	
	\newtheorem{corollary}[theorem]{Corollary}
	\newtheorem{claim}{Claim}[theorem]
		
	\newtheorem{example}{Example}
	\theoremstyle{remark}

	\title{On outer-connected domination for graph products}
		\author{ M. Hashemipour$^1$,  M. R. Hooshmandasl$^2$, A. Shakiba$^3$ \\
			\footnotesize{$^{1,2}$Department of Computer Science, Yazd University, Yazd, Iran. }  \\
			\footnotesize{$^{1,2,3}$The Laboratory of Quantum Information Processing, Yazd University, Yazd, Iran.}  \\
			\footnotesize{$^{3}$Department of Computer Science, Vali-e-Asr university of Rafsanjan}\\
			\footnotesize{e-mail: $^1$mhashemi@stu.yazd.ac.ir, $^2$hooshmandasl@yazd.ac.ir, $^3$ali.shakiba@vru.ac.ir.
			}}

	\date{}
	
\begin{document}
			\maketitle
	\begin{abstract}
		An \emph{outer-connected dominating set} for an arbitrary graph $G$ is a set $\tilde{D} \subseteq V$
		   such that
		    $\tilde{D}$ is a dominating set
		     and  the induced subgraph $G [V \setminus \tilde{D}]$ be connected. 
				In this paper, we focus on the outer-connected domination number of the product of graphs.
					We investigate the existence of outer-connected dominating set in lexicographic  product and Corona of two arbitrary graphs, 
					 and we present upper bounds for outer-connected domination number in lexicographic and Cartesian product of graphs. Also, we establish an equivalent form of the Vizing's conjecture for outer-connected domination number in lexicographic and Cartesian product as  $\tilde{\gamma_c}(G \circ K)\tilde{\gamma_c}(H \circ K) \leq \tilde{\gamma_c}(G\Box H)\circ K$.
		 Furthermore,	we study the outer-connected domination number of the direct product of finitely many complete graphs. \\
		 \textbf{Keywords:} Outer-connected domination; Cartesian product; Lexicographic  product;  Corona product; Direct product; Vizing's conjecture.

	\end{abstract}	
	
		\section{Introduction \& preliminary}

 Domination and its variations in graphs are a well studied topic in the literature, e.g.\cite{haynes1998fundamentals} gives a survey on the topic.  The concept of outer-connected domination number, as a variant of graph domination problem, is introduced by Cyman \cite{cyman2007outer1connected} and is further studied by others in \cite{akhbari2013outer,keil2013computing}. The outer-connected domination problem is $NP$-complete for
arbitrary graphs \cite{cyman2007outer1connected}.
A set $\tilde{D} \subseteq V$ of a graph $G=(V,E) $ is called an \emph{outer-connected dominating set} for $G$ if (1) $\tilde{D}$ is a dominating set for $G$, and (2)  $G [V \setminus \tilde{D}]$, the induced subgraph of $G$ by $V \setminus \tilde{D}$, is connected. The minimum size among all outer-connected dominating sets of $G$ is called the \emph{outer-connected domination number} of $G$ and is denoted by $\tilde{\gamma}_c(G)$\cite{cyman2007outer1connected}. 

The problem of finding a minimum sized outer connected dominating set has applications in computer networks.
For example consider a client-server architecture based network in which any client must be able to communicate to one of the servers. Since overload of the servers is a bottleneck
in such a network, every client must be able to communicate to another client
directly without interrupting any  server.  The smallest group of servers
with these properties is a minimum outer-connected dominating set for the
graph representing the computer network\cite{panda2014algorithm}.
\\
 We make use the following result related to the outer-connected
 domination number in this paper.
 
 \begin{theorem}\cite{cyman2007outer1connected}
 	If $G$ is a connected graph of order $n$, then 	
 	\begin{equation*}
 	\tilde{\gamma_c} (G) \le n- \delta(G).
 	\end{equation*}	
 \end{theorem}

Graphs are basic combinatorial structures and products of structures are a
fundamental construction in graph theory. Such construction is a challenging problem and has many applications.
 In graph theory there are three
fundamental graph  products,  namely the Cartesian product, the direct product, and the strong product, 
 each with its own set of applications and theoretical interpretations.
 Computer science is one of the
many fields in which graph products are becoming commonplace. As one specific example,
one can mention the load balancing problem for massively parallel computer architectures\cite{hammack2011handbook}. In addition to large
networks such as the graph of the which has Internet several hundred million hosts, can be efficiently
modeled by subgraphs of powers of small graphs with respect to the direct product. This
is one of the many examples of the dichotomy between the structure of products and that of
their subgraphs\cite{hammack2011handbook}. The classification also leads to a other two products worthy of special attention, the
  lexicographic and the Corona products\cite{hammack2011handbook}.

  The study of domination number in product graphs has a long history. Back in 1963, Vizing \cite{vizing1963cartesian} posed a conjecture, which
  is  main open problem in
  graph domination, concerning the domination number of the Cartesian product graphs
 \begin{equation*}
  		    \gamma(G)\gamma(H) \le \gamma(G\Box H).
 		\end{equation*}
  For a survey of domination in Cartesian products,
 an interested reader can consult\cite{hartnell1998domination} for more information. 
  Gravier and Khelladi \cite{gravier1995domination} posed an analogous conjecture for direct product graphs, namely
 \begin{equation*}
 \gamma(G)\gamma(H) \le \gamma(G\times H).
 \end{equation*}

 Domination number of direct products of certain graphs has exact values, for instance, 
 the products of two paths, the product of a path and a complement of a path, the product of K2 and a tree,
 bipartite graph and an odd cycle \cite{gravier1995domination,klobuvcar1999domination,klobuvcar1999domination1}.
 In 2010, Gasper Mekis\cite{mekivs2010lower} gave a lower bound for the domination number of a direct product  and proved that this bound is sharp. Also, he studied 
 the domination number of the direct product of finitely many complete graphs.

 For the lexicographic product graphs, various types of domination were investigated
in the literature, including domination \cite{vsumenjak2012roman,nowakowski1996associative}, 
 total domination \cite{zhang2011domination}, rainbow domination \cite{vsumenjak2013rainbow}, Roman domination \cite{vsumenjak2012roman}, and restrained domination \cite{zhang2011domination}.  
  Other various types of dominating sets for products
  of graphs were intensively investigated in \cite{henning2013total,yero2014efficient}. However, outer-connected dominating sets for products 
  of graphs has not been investigated. So,   we investigate the topic  in this paper by studying the outer-connected dominating  sets and related notions
  in the lexicographic product, the direct product, the Cartesian product and the Corona product graphs.

  For notation and graph theory terminology, we in general follow \cite{haynes1998fundamentals}.	Let $G = (V,E)$ be a graph with vertex set $V$ and edge set $E$ of  order $|V|$ denoted by $n$ and size $|E|$ denoted by $m$.  We also use $V(G)$ and $E(G)$ to denote the vertex set and 
  edge set for a graph $G$.
  Let  $v$ be a vertex in $V.$ The open neighborhood of $v$ is denoted by $N_G(v)$ and is defined as $ \{u \in V : \{u,v\} \in
  E(G)\}$. Similarly, the closed neighborhood of $v$ is denoted by $N_G[v]$ and is defined as $\{v\} \cup N_G(v)$.  Whenever the graph $G$
  is clear from the context, we simply write $N(v)$ to denote $N_G(v)$.
  For a	set $S \subseteq V$, its open neighborhood is the set $N(S) = \cup_{v\in S} N(v)$ and its closed
  neighborhood is the set $N[S] = N(S) \cup S$. A subset $S \subseteq V$ is a \emph{dominating set} of
  $G$ if every vertex not in $S$ is adjacent to a vertex in $S$. The \emph{domination number}
  of $G$, denoted by $\gamma(G)$, is the minimum cardinality among all dominating sets of $G$. A
  dominating set $S$ is called a $\gamma(G)-set$ of $G$ if $|S| =\gamma(G)$. A dominating set $S$ in a graph with no isolated vertex is called a \emph{total dominating set } if 
  the induced subgraph $G[S]$ has no isolated vertex. The \emph{total domination number}
  of $G$, denoted by $\gamma_t(G)$, is the minimum cardinality among all total dominating sets of $G$.
  A total dominating set $S$ is called  a $\gamma_t(G)-$set of $G$ if $|S| = \gamma_t(G)$.

   We in general follow the product of graphs in \cite{hammack2011handbook}.
  The lexicographic product of two graphs G and H, denoted by $G \circ H$, is the
   	graph with vertex set equal to $ V (G) \times V (H)$ such that two vertices $(u_1, u_2)$ and $(v_1, v_2)$ are connected by an edge if either
   	$\{u_1,v_1\} \in E(G)$ or $u_1 = v_1$ and $\{u_2,v_2\}\in E(H)$.
  The Corona product of two graphs G and H, denoted by $G \circ_c H$, is the graph obtained
  	by taking one copy of G  and n copies of H, where G has n vertices, and joining the
  	$i^{th}$ vertex of G to every vertex in the 	$i^{th}$ copy of H. For every $x \in V (G)$, we
  	denote the copy of H whose vertices are attached to the vertex
  	$x$ in G  by $H^x$.
   The Cartesian product of two graphs $G$ and $H$, denoted by $G \Box H$, is the
   graph with vertex set equal to $ V (G) \times V (H)$ such that two vertices $(u_1, u_2)$ and $(v_1, v_2)$ are connected by an edge if either
   $\{u_1,v_1\} \in E(G)$ and $u_2 = v_2$, or $u_1 = v_1$ and $\{u_2,v_2\}\in E(H)$.   	 	
   		For graphs G and H, the direct product denoted by $G \times H$ (also known as the tensor product, cross product, cardinal product  and 	categorical product), is the graph with vertex set equal to $V (G) \times V (H)$ such that two vertices $(u_1, u_2)$ and $(v_1,v_2)$ are connected by an edge
   		if and only if $\{u_1,v_1\} \in E(G)$ and $\{u_2,v_2\} \in E(H)$.
   		
The rest of the paper is organized as follows: In Section 2, we characterize the outer-connected domination of  lexicographic
products of graphs by constructing minimum sized ones. In Section 3,  we investigate the outer-connected domination of Corona
products of graphs. In Section 4,  an upper bound for the outer-connected domination number in the Cartesian product graphs is defined. Finally, in Section 5, we study 
the outer-connected domination number of the direct product of finitely many complete graphs.

\section  {Outer Connected Domination in the Lexicographic
	Product of Two Connected Graphs}

	
    In this section we must determine the outer-connected domination number in the lexicographic product of two graphs.  To this aim, we first prove the following three
   lemmas.

\begin{lemma}
	\label{th2}
	Let G and H be two graphs. Then, $\gamma(G) \le \gamma(G \circ H)$.
\end{lemma}
\begin{proof}
	Let $T = \set{(x_1,y_1),(x_2,y_2),\cdots, (x_k,y_k)}$ be the dominating set for $G \circ H$. Then for every vertex $(a,b)\in G \circ H, $ there exists a vertex $(x_i,y_i) \in T$ such that $\set{(x_i,y_i),(a,b)}\in E(G \circ H)$. That means either  $\set{x_i,a }\in E(G)$ or  $x_i=a$ and $ \set{y_i,b} \in E(H)$. 	In both cases, the set $T^\prime = \set{x_i \mid (x_i,y_i) \in T}$ is a dominating set for $G$ and its cardinality is less than or equal  to the cardinality of the set $T$.
	
\end{proof}

\begin{lemma}
	\label{th3}
	Let $\gamma(H) \neq 1$ and $T$ is a dominating set for $G \circ H$  and $(x_i,y_i) \in T$. then there exists at least one vertex  of the form $(x_j,v) \in T$ 
	such that either $\set{x_i,x_j} \in E(G)$  or $x_i = x_j$ and $v \neq y_i$.
\end{lemma}
\begin{proof}
	It is clear that the vertex  $(x_i,y_i) \in T$ cannot cover all of the vertices $(x_i,u) \in V(G \circ H)$ because  $\gamma(H) \neq 1$. So, the  vertex $(x_i,u)$ is covered by the vertex $(x_j,v) \in T$ such that either $\set{x_i,x_j} \in E(G)$  or $x_i = x_j$ and $v \neq y_i$.
\end{proof}

\begin{lemma}
	\label{th4}
	
	Let G and H be two graphs such that $\gamma(H)\neq 1$. Then, there exists  a  total dominating set for G with   cardinality less than or equal to $\gamma(G \circ H)$. 
\end{lemma}
\begin{proof}
	According to Lemma \ref{th2}, the set $T^\prime $ is a dominating set for G. On the other hand, we have $\gamma(H)\neq 1$, so by Lemma \ref{th3}, for all vertices $(x_i,v) \in  T$, there exists at least one node with the necessary conditions. Then,  for all vertices $x_i \in T^\prime$, we also add one of its neighbors to $T^\prime$ which would be a total dominating set for $G$ of cardinality less than or equal to $|T|$. So, $\gamma_t(G) \le \gamma(G \circ H).$
	\end {proof}


		\begin{theorem}
			\label{th5}
			Suppose that $G$ and $H$ are two connected graphs. Then, we have
		
		\begin{equation}
	\tilde {\gamma_c}(G \circ H) = \begin{cases} 1 & \text{if} ~~ \gamma(G)=1 ~~ \text{and} ~~ \gamma(H)=1, \\ 	2 & \text{if} ~~ \gamma(G) = 1 ~~ \text{and} ~~ \gamma(H) \neq 1, \\ 	\gamma(G) & \text{if} ~~ \gamma(G) \neq 1 ~~ \text{and}  ~~ \gamma(H)=1, \\	\gamma_t(G) & \text{if} ~~ \gamma(G) \neq 1 ~~ \text{and}  ~~ \gamma(H)\neq 1 . \end{cases}	
	\end{equation}
	
\end{theorem}
	
		
			\begin{proof}
				The proof is by construction:
				\begin{quote}\begin{description}
					\item[Case 1: $\gamma(G) = \gamma(H) = 1$:] 
					 Without loss of generality, assume that $\set{x}$ and $\set{y}$ are minimum dominating sets for $G$ and $H$, respectively. We claim that $\set{(x,y)}$ is an outer-connected dominating set for $G \circ H$.					 
					 For every vertex $u \in V(G)$ and $u \neq x $, it is clear that $\set{x,u} \in E(G) $. So, by the definition  of the lexicographic product, for all vertices $v \in V(H)$, we have
					  \begin{equation}
					  \set{(x,y),(u,v)} \in E(G \circ H).
					   \end{equation} 					  
					  
					On the other hand,  by the definition  of lexicographic product and the fact that $\set{y}$ is a dominating set for $H$, we have
				
				\begin{equation}
				\set{(x,v),(x,y)} \in E(G \circ H).
				\end{equation}
				 
					Therefore, the set $\set{(x,y)}$ is a dominating set for $G \circ H$.  Next, we need to show that the induced graph $ G \circ H \setminus (x,y)$ is connected.
					 Let				
					  $ V(G) = \set{v_1,v_2,\cdots,v_{n-1},x}$ and 					   
					  $ V(H) = \set{u_1,u_2,\cdots,u_{n-1},y}.$					
					So, it suffices to show that for all $(a,b) , (c,d) \in V(G \circ H) \setminus (x,y),$
				    	there exists a path from (a,b) to (c,d) in $G \circ H$ which  does not pass through the vertex (x,y). To this end, there are three cases to consider:
					\begin{quote}\begin{description}
					  
					  	\item[Case 1-a: $a \neq x~~ \text{and}~~c \neq x$:]
					  	
					  	 By the definition  of the lexicographic product and given that the set $\set{x}$ is a dominating set for $G$, we have
					  	 
					  	 \begin{equation*}
					  	\set {(x,v),(a,b)} \in E(G \circ H),
					  	 \end {equation*}
					  	 \begin{equation*}
					  	\set {(x,v),(c,d)} \in E(G \circ H).
					  	 \end{equation*}
					  	 So, there exists a  path of length two between vertices $(a,b)$ and $(c,d)$ which passes through the vertex $(x,v)$.
					  	\item[Case 1-b: $a = x~~ \text{and} ~~c \neq x$:]
					  	 By the definition  of the lexicographic product and given that the set $\set{x}$ is a dominating set for $G$, it is clear that  for all $b\in H$ where $b\neq y$, we have
					  	 
					  	 \begin{equation*}
					  	\set {(x,b),(c,d)} \in E(G \circ H).            
					  	 \end {equation*}
					  	 
					  	 \item[Case 1-c: $a = c = x,~~b \neq y~~\text{and}~~d \neq y$:]
					  	For vertex $(t,u)$ where $t \neq x $, we have
					  	 
					  	 \begin{equation*}
					  \set{(a,b),(t,u)} \in E(G \circ H), ~~ ~~~~~             
					  	 \end {equation*}
					  	  \begin{equation*}
					  	  \set{(c,d),(t,u)} \in E(G \circ H). ~~ ~~~~~             
					  	    \end {equation*}
					  	  
					 So, there exists a path of length two between vertices (a,b) and (c,d) which passes through the vertex (t,u).  	
			\end{description}\end{quote}
				Finally, because the minimum possible size for any dominating set is one, this is clearly a minimum outer-connected dominating set for $G [H]$.
			
			\item[Case 2: $\gamma(G) =  1 ~~ \text{and} ~~\gamma(H) \neq 1$:]
			     Assume that the set $\set{x}$ is a dominating set for G. Similar to  the previous case, for any vertex $u\neq x$ where $u \in V(G) $, all the vertices $(u,v)$ are adjacent to the vertex $(x,y)$ for  $y,v \in V(H)$.
			     Also, each vertex $( u^\prime , v^ \prime  )$ where $u^\prime \neq x$ dominates all the vertices of the form $(x,v)$. So, the set 
			     \begin{equation*}
			     \tilde D (G \circ H) = \set{(x,y),(u^\prime , v^\prime)},
			     \end{equation*}
			     is a dominating set for $G \circ H$.
			     Then, we have
			      \begin{equation*}
			      \tilde \gamma_c(G \circ H) = 2.
			      \end{equation*}
			      
			      Finally, we need to show that the induced graph $ G \circ H \setminus  \tilde D (G \circ H)$ is connected. 
			    To do so, we can apply the same method as in the previous case except that in case 1-c, we need to consider the constraint $(t,u) \neq (u^\prime , v^\prime)$.\\

			 \item[Case 3: $\gamma(G) \neq  1 ~~\text{and} ~~\gamma(H) = 1$:]     
			      
			      	Suppose that the set 	$\mathcal{S}  =  \set{x_1,x_2,\cdots,x_m}$		      	is a  minimum cardinality dominating set for G.  By the definition  of the  lexicographic product and the fact that $\set{y}$ is a dominating set of $H$, we have
			      	
			      		\begin{equation*}
			      	\mathcal{S^\prime}  =  \set{(x_1,y),(x_2,y),\cdots,(x_m,y)},			
			      		\end{equation*}
			      		
			       is a dominating set for $G \circ H $ and 	$\tilde {\gamma_c}(G \circ H) = \gamma(G) $ since $\mid \mathcal{S^\prime}\mid = \mid \mathcal{S} \mid$.\\
			      	Next, we consider vertices $(a,b) , (t,p) \in V(G \circ H) \setminus \mathcal{S^\prime} $. We know that G is a connected graph, so there exists a path from vertex a to vertex t in graph G which is denoted by 
			      	\begin{equation*}
			      	a \rightarrow a_1 \rightarrow a_2 \rightarrow \cdots \rightarrow a_k \rightarrow t,  
			      	\end{equation*}
			      	
			      	where $ a_1 , a_2 , \cdots , a_k \in V(G).$ So, there exists a path from $(a,b)$ to $(t,p)$ in $G \circ H \setminus \mathcal{S^\prime}$ denoted by
			      	\begin{equation*}
			      	(a,b) \rightarrow (a_1,u^\prime_1) \rightarrow (a_2,u^\prime_2) \rightarrow \cdots \rightarrow (a_k,u^\prime_k) \rightarrow (t,p).     
			      	\end{equation*}				 
			      	
			      	So, the induced graph $V(G \circ H) \setminus \mathcal{S^\prime}$ is connected.
			      	
			      	In follow, we show that the set $ \mathcal{S^\prime}$ is minimum.			      	
			       	Let $S^\prime$ be an outer-connected dominating set for  $G \circ H$. Then,
			      	\begin{equation}
			      	\label{eq1}
			      	\tilde{\gamma_c}(G \circ H) \le \mid S^\prime \mid = \gamma(G). 
			      	\end{equation}
			      	
			      	Now, suppose that $S^*$ is a minimum cardinality outer-connected dominating set for $G \circ H$. So, by Lemma \ref{th2},    there is a  dominating set for G with cardinality $\mid S^* \mid$ which we call it $T^*$. So, we have
			      	\begin{equation}
			      	\label{eq2}
			      	\gamma(G) \le \mid T^* \mid \le \mid S^* \mid = \tilde{\gamma_c}(G \circ H), 
			      	\end{equation}
			      	which leads to
			      	$\tilde{\gamma_c}(G \circ H) = 	\gamma(G)$
			      	by Equations \ref{eq1} and \ref{eq2}.

			\item[Case 4: $\gamma(G) \neq  1 ~~ \text{and} ~~\gamma(H) \neq 1$:]
			
			Suppose that the set $\mathcal{S}  =  \set{x_1,x_2,\cdots,x_t}$
			is a minimum cardinality total dominating set for G. For every vertex $x_i \in \mathcal{S},$ the set of vertices dominated by the vertex  $x_i$  is denoted as $S_i$.
			Since the set  $\mathcal{S}$ is a total dominating set for G, we have 
			\begin{equation*} 
			\bigcup_{i=1}^t S_i = V(G), 			
			\end{equation*}
			and for all $x_j \in \mathcal{S}$, there exists a vertex $x_i \in \mathcal{S}$ such that  $x_j \in S_i$. So, the vertex
					
				  $(x_i,v)$ dominates all the vertices of the form $(x_j,v^\prime)$ in $G \circ H$. So, the set 
				\begin{equation*}
				\mathcal{S^\prime}  =  \set{(x_1,v),(x_2,v),\cdots,(x_t,v)},			
				\end{equation*}
				
			dominates all the vertices $(a,b) \in V(G \circ H)$ where $a \in 	\bigcup_{i=1}^t S_i $. Therefore, the set $	\mathcal{S^\prime} $ is a dominating set for $G \circ H $.
			
			 Next, we consider the vertices $(x,y) , (t,p) \in V(G \circ H) \setminus \mathcal{S^\prime} $. We know that G is a connected graph so, there exists a path from vertex x to vertex t in graph G which is denoted by 
			 \begin{equation*}
			 x \rightarrow a_1 \rightarrow a_2 \rightarrow \cdots \rightarrow a_k \rightarrow t, 
			 \end{equation*}
			 
			          where $ a_1 , a_2 , \cdots , a_k \in G.$ So, there exists a path from $(x,y)$ to $(t,p) $ in $G \circ H \setminus \mathcal{S^\prime}$ denoted by
				 \begin{equation*}
				 (x,y) \rightarrow (a_1,u^\prime_1) \rightarrow (a_2,u^\prime_2) \rightarrow \cdots \rightarrow (a_k,u^\prime_k) \rightarrow (t,p).     
				 \end{equation*}				 
				 
				So, the induced graph $V(G \circ H) \setminus \mathcal{S^\prime}$ is connected.\\				
			Finally, to complete the proof, it is necessary to show that  the  dominating set obtained for $G \circ H$   is minimum. To do so, we can apply the same method as in the previous case. 
		
				Let $S^\prime$ be an outer-connected dominating set for  $G \circ H$. Then,
				\begin{equation}
				\label{eq3}
				\tilde{\gamma_c}(G \circ H) \le \mid S^\prime \mid = \gamma_t(G). 
				\end{equation}
				Now, suppose $S^*$ is a minimum cardinality outer-connected dominating set for $G \circ H$. So, by Lemma \ref{th4},    there is a total dominating set for G with cardinality $\mid S^* \mid$ which we call it $T^*$. So, we have
				\begin{equation}
				\label{eq4}
				\gamma_t(G) \le \mid T^* \mid \le \mid S^* \mid = \tilde{\gamma_c}(G \circ H), 
				\end{equation}
				which leads to
				$\tilde{\gamma_c}(G \circ H) = 	\gamma_t(G)$
				by Equations \ref{eq3} and \ref{eq4}.
					
				\end{description}\end{quote}
			\end{proof}

By using the above theorem, we can write an equivalent form of the Vizing's conjecture as follow.

\begin{theorem}
	Let  $G$, $H$ and $K$ be graphs such that $\gamma(G) \neq 1$, $\gamma(H) \neq 1$ and $\gamma(K) = 1$. The Vizing's conjecture is true if and only if
	
	\begin{equation*}
	\tilde{\gamma_c}(G \circ K)\tilde{\gamma_c}(H \circ K) \le \tilde{\gamma_c}(G\Box H)\circ K.
	\end{equation*}
	
\end{theorem}

\begin{proof}
	According to Vizing's conjecture for all graphs $G$ and $H$,
	\begin{equation*}
	\gamma(G)\gamma(H) \le \gamma(G\Box H).
	\end{equation*}			 

	By  using Theorem \ref{th5} we get the following inequality:
	\begin{equation*}
	\tilde{\gamma_c}(G \circ K)\tilde{\gamma_c}(H \circ K) = \gamma(G)\gamma(H) 
	\le \gamma(G\Box H) \le \tilde{\gamma_c}(G \Box H)
	= \tilde{\gamma_c}((G\Box H)\circ K),
	\end{equation*}
therefore
	\begin{equation*}
	\tilde{\gamma_c}(G \circ K)\tilde{\gamma_c}(H \circ K)\leq \tilde{\gamma_c}((G\Box H)\circ K).
	\end{equation*}
Conversely, we consider that the inequality  $
\tilde{\gamma_c}(G \circ K)\tilde{\gamma_c}(H \circ K)\leq \tilde{\gamma_c}((G\Box H)\circ K)
$ is true. By assumptions of the theorem, $\gamma(G) \neq 1$, and that  $\gamma(G\Box H) \geq \gamma(G)$ we have $\gamma(G\Box H)\neq 1$. So, by using Theorem \ref{th5} we get
$$\gamma(G)\gamma(H)=
\tilde{\gamma_c}(G \circ K)\tilde{\gamma_c}(H \circ K)\leq \tilde{\gamma_c}((G\Box H)\circ K)=\gamma(G\Box H),
$$
or
$$\gamma(G)\gamma(H)\leq \gamma(G\Box H).
$$
\end{proof}

							


			
		\begin{lemma}
		 If $H=K_1$, then we have $\tilde {\gamma_c} (G \circ H) = \tilde{\gamma_c} (G)$.
		\end{lemma}	
		\begin{proof}
		  It is easy to verify, since  for every $x \in V(G)$, there exists exactly one vertex $(x,y) \in V(G \circ H)$ where $y$ is the only vertex in H and for every $\set{x,u} \in E(G) $, there exists exactly one edge connecting  $(x,y)$ and $(u,y)$ in $G \circ H$. 
		      
		\end{proof}
		‌‌ 
			\begin{lemma}
				If $G=K_1$, then we have $\tilde {\gamma_c} (G \circ H) = \tilde{\gamma_c} (H)$.
			\end{lemma}	
				\begin{proof}
					It is clear since for every $y \in V(H)$, there exists exactly one vertex $(x,y) \in V(G \circ H)$ where $x$ is the only vertex in G and for every $\set{v,y} \in E(G) $, there exists exactly one edge connecting  $(x,y)$ and $(x,v)$ in $G \circ H$. 
					
				\end{proof}

	By Theorem \ref{th5} outer-connected domination number of $G \circ H$ denoted by $\tilde {\gamma_c} (G \circ H)$ depends on the values $\gamma (G)$ and $\gamma_t (G)$. We know that it is NP-hard to compute the domination number and the total domination number. So, We apply the following Lemma to determine the upper bound for $\tilde {\gamma_c} (G \circ H)$.
		
	\begin{lemma}
		\label{lem1}
		If G is a connected graph of order n and H is a graph of order m, then we have
		\begin{equation}
		\tilde{\gamma_c} (G \circ H) \le mn-(\delta(G) + \delta(H)).
		\end{equation}
	\end{lemma}
	\begin{proof}
		By the defination of the Lexicographic products, $G \circ H$ is a graph of order $mn$ and $\delta(G \circ H) = \delta(G) + \delta(H)$. Therefore we have  
		\begin{equation}
		\tilde{\gamma_c} (G \circ H) \le mn-(\delta(G) + \delta(H)),
		\end{equation}
		by Theorem $5$ in \cite{cyman2007outer1connected}.
	\end{proof}
		The following is a tight example for Lemma \ref{lem1}.
	\begin{example}

		Let $G$ and $H$ be two graphs with $V(G) =\{p\}$ , $V(H) =\{x,y\}$ and $E(H)=\{\{x,y\}\}$.  Then,  we have
		\begin{equation}
    	V(G \circ H) = \{(p,x) , (p,y)\},
    	\end{equation} 
    	\begin{equation}    		
	    E(G \circ H) = \set{\{(p,x),(p,y)\}},
	    \end{equation}
	    
	    	\begin{equation}
	    	  \tilde{\gamma_c}(G \circ H) = 1,
	    	\end{equation}
	    and	
	   		\begin{equation}
	   		 mn -(\delta(G) + \delta(H)) = 1*2 -(0+1) = 1.
	   		 \end{equation}
	\end{example}

	\section{Outer Connected Domination in the Corona
		Product of Two Graphs}

		\begin{lemma}
			\label{corona}
			Suppose that $G$ is a connected graph and $\tilde{D} \subseteq V(G)$ is an outer-connected dominating set for
			$G$. If $u \in \tilde{D}$ is a cut vertex, then all the vertices $v \in V \setminus \tilde{D}$  belong to exactly a single component of  $V \setminus \{u\}$.  
		\end{lemma}
		
		\begin{proof}
			Let  $u \in \tilde{D}$ be a cut vertex and $c_1,c_2,\dots , c_m$ be components of the induced subgraph  $G[V \setminus \{u\}]$. Suppose that there exist arbitrary vertices
			$x,y \in V \setminus \tilde{D}$ such that $x \in c_i$ and $y \in c_j$ for $i \neq j , 1\le i,j \le m$.
			Thus, there exists no path from $x$ to $y$ in $G[V \setminus \{u\}]$, and this is a contradiction to the assumption that  $\tilde{D}$ is an outer-connected dominating set. So, 
			$x$ and $y$ are certainly in the same component of  $G[V \setminus \{u\}]$. Since the vertices $x$ and $y$ are chosen arbitrarly, then the theorem is proven.
		\end{proof}
		\begin{corollary}
				Suppose that $G$ is a connected graph and $\tilde{D} \subseteq V(G)$ is an outer-connected dominating set for
				$G$. Let  $u \in \tilde{D}$ be a cut vertex  and $c_1,c_2,\dots , c_m$ be components of the induced subgraph  $G[V \setminus \{u\}]$. Then, all the  vertices in its $m-1$ components are included in $\tilde{D}$.
			
					\end{corollary}

		\begin{theorem}
			\label{th6}
			Let G be a connectecd graph and H is an arbitrary graph. The set $\tilde{D} \subset V(G \circ_c H)$ is an outer-connected domination for $G \circ_c H$ if and only if 
			\begin{equation}
			\tilde{D} = \cup_{x \in V(G)}(D(H^x)),
			\end{equation}
			where D is a minimum dominating set of graph $H^x$, and $H^x$ is the copy of graph H whose vertices are attached to the vertex
			$x$ in G. 		
		\end{theorem}
		\begin{proof}
			If the set $\tilde{
				D} \subseteq V(G \circ_c H)$ is an outer-connected domination for $G \circ_c H$ and $x \in \tilde{D}$, then x is not a cut vertex. Otherwise, according to Theorem \ref{corona} and  by the assumption that the graph $(G \circ_c H) \setminus {v}$ has m components, all  the vertices in its $m-1$ components are included in the  outer-connected dominating set. As a result, the set 
			$\tilde{D}$ is not minimum. \\
			Since every vertex $v \in V(G)$	 is a  cut vertex in  $G \circ_c H$, then none of the vertices $v \in G$ are  in $\tilde{D}$.
			Therefore, by the definition of the  Corona product and the obtained fact that  $v \in G$ is not  in $\tilde{D}$, we have 
			\begin{equation}
			\tilde{D} = \cup_{x \in V(G)}(D(H^x)).
			\end{equation}

			To prove the converse of the theorem, suppose that	$\tilde{D} = \cup_{x \in V(G)}(D(H^x))$. It is clear that $D(H^x)$ is a  dominating set for $H^x \cup \set{x}$, then $\cup_{x \in V(G)}(D(H^x))$ is a  dominating set for $G \circ H$. On the other hand, G is a connected graph and all of its vertices $u \in H^x \setminus D(H^x)$ are connected to $x$. So, $\tilde{D}$ is an outer-connected dominating set for $G \circ_c H$. 
			Eventually, it is clear that
			\begin{equation}
			\mid D(H^x) \mid \le \mid \tilde{D}(H^x) \mid.  
			\end{equation} 
			So, the set $\tilde{D}$ is minimum.  	
		\end{proof}
		\begin {corollary}
		Let G be a connected graph 
		and H be an arbitrary graph. Then, we have
		\begin{equation}
		\tilde \gamma_c(G \circ_c H) = |V(G)| \gamma(G).
		\end{equation}
	\end{corollary}
	\begin{proof}
	The proof is clear from Theorem \ref{th6} and the defination of the corona product.	
		
	\end{proof}
	
		\section{Outer Connected Domination in the Cartesian
						Product of Two Graphs}

					 	In this section, we present an upper bound for outer-connected domination number in Cartesian product graphs.

					 

					\begin{theorem}
						\label{cart}
						For any graphs G and H, we have
						\begin {equation}
						\tilde{\gamma_c}(G \Box H) \le \tilde{\gamma_c}(G) \times |V(H)|.
					\end{equation}
				\end{theorem}

				\begin{proof}
						We first show the following  claim and using it, we have
						
							\begin{equation}
							\tilde{\gamma_c}(G \Box H) \le |T| = | \tilde{D} | \times |V(H)| = \tilde{\gamma_c}(G) * |V(H)|
							\end{equation}		
							
						\begin{claim}
							\label{cl1}
					   Let $\tilde{D}$  be an outer-connected dominating set of $G$. Then, $T=\tilde{D} \times V(H)$ 
					  is an outer-connected dominating set for $G \Box H$. 
					  \end {claim}
					 \begin{proof}
					 	
					  Let $(a, b) \in V(G\Box H)\setminus T$, then we have $a \notin\tilde{D}$. Since $\tilde{D}$ is an Outer-connected dominating set for $G$, then there
						exists a vertex $x$ in $\tilde{D}$ such that $\{a,x\} \in E(G)$. On the other hand, $b \in V(H)$ so we have $(x, b) \in T$. Then, according to the defination of the Cartesian product, the vertex $(a,b)$ is adjacent to vertex $(x, b)$ in $G \Box H$ which means that the vertex $(a,b)$ is dominated by the vertex $(x, b) \in T$.
						
						Now, consider two vertices $(a,b),(x,y) \in V(G \Box H ) \setminus T$. If $a = x$,  then there exist a path from vertex $b$ to vertex $y$ in graph $H$
						in the form  
						
							\begin{equation*}
							b \rightarrow t_1 \rightarrow t_2 \rightarrow \cdots \rightarrow t_k \rightarrow y, 
							\end{equation*}
							where $ t_1 , t_2 , \cdots , t_k \in H$  since $H$ is a connected graph. So, there exists a path from $(x,y)$ to $(a,b) $ in $(G \Box H) \setminus T$ denoted by
							\begin{equation*}
							(a,b) \rightarrow (a,t_1) \rightarrow (a,t_2) \rightarrow \cdots \rightarrow (a,t_k) \rightarrow (x,y).     
							\end{equation*}								
						If $a \neq x$ and $ y = b$, since $\tilde{D}$ is an outer-connected dominating set of  graph $G$  and $a,x \notin\tilde{D}$, then there exists a path from vertex $a$ to vertex $x$ in the graph $G \setminus \tilde{D}$
						which is denoted by
						
						\begin{equation*}
						a \rightarrow t_1 \rightarrow t_2 \rightarrow \cdots \rightarrow t_k \rightarrow x, 
						\end{equation*}
						where $ t_1 , t_2 , \cdots , t_k \in G \setminus \tilde{D}.$ So, there exists a path from $(x,y)$ to $(a,b) $ in $(G \Box H) \setminus T$ denoted by
						\begin{equation*}
						(a,b) \rightarrow (t_1,b) \rightarrow (t_2,b) \rightarrow \cdots \rightarrow (t_k,b) \rightarrow (x,y).     
						\end{equation*}

							If $a \neq x$ and $ y \neq b$, then there exists a path from vertex $a$ to vertex $x$ in graph $G \setminus \tilde{D}$
							which is denoted by
							
							\begin{equation*}
							a \rightarrow t_1 \rightarrow t_2 \rightarrow \cdots \rightarrow t_k \rightarrow x,
							\end{equation*}
							where $ t_1 , t_2 , \cdots , t_k \in G \setminus \tilde{D}$ since $\tilde{D}$ is an outer-connected dominating set of  graph $G$  and $a,x \notin\tilde{D}$. So, there exists a path from $(a,y)$ to $(x,y) $ in $(G \Box H) \setminus T$ denoted by
							\begin{equation*}
							(a,y) \rightarrow (t_1,b) \rightarrow (t_2,b) \rightarrow \cdots \rightarrow (t_k,b) \rightarrow (x,y).     
							\end{equation*}

						On the other hand, since $H$ is a connected graph, then there exists a path from the vertex $b$ to the vertex $y$ in graph $H$
						of the form  
						
						\begin{equation*}
						b \rightarrow t_1 \rightarrow t_2 \rightarrow \cdots \rightarrow t_k \rightarrow y, 
						\end{equation*}
						where $ t_1 , t_2 , \cdots , t_k \in H.$ So, there exists a path from $(a,b)$ to $(a,y) $ in $(G \Box H) \setminus T$ denoted by
						\begin{equation*}
						(a,b) \rightarrow (a,t_1) \rightarrow (a,t_2) \rightarrow \cdots \rightarrow (a,t_k) \rightarrow (a,y).     
						\end{equation*}				 	
							
						Therefore, there exists a path from $(a,b)$ to $(x,y)$ in $(G \Box H) \setminus T$  which passes through the vertex $(a,y)$.

						So, the induced graph $V(G \Box H) \setminus T$ is connected.
					and hence, $T$ 
					is an outer-connected dominating set for $G \Box H$. 
						\end{proof}
				
				\end{proof}

	The following tight example shows that the bound given in Theorem \ref{cart} is sharp.

	\begin{example}
		Let $G$ and $H$ be two graphs with $V(G) =\{a,b,c\}$ , $V(H) =\{x,y\}$ , 
		$E(G)=\{\{a,b\},\{a,c\},\{b,c\}\}$ and $E(H)=\{\{x,y\}\}$.  Then,  we have
		\begin{equation}
		V(G \Box H) = \{(a,x) , (a,y) , (b,x) , (b,y) , (c,x) , (c,y)\},
		\end{equation} 
		\begin{align}
		\begin{aligned}
			E(G \circ H) = \{\{(a,x),(a,y)\}, \{(b,x),(b,y)\} , \{(c,x),(c,y)\} 
				\\
				, \{(a,x),(b,x)\} , \{(a,x),(c,x)\} , \{(a,y),(b,y)\}					 
			\\
				, \{(a,y),(c,y)\} , \{(b,x),(c,x)\} , \{(b,y),(c,y)\}\}.
		\end{aligned}
		\end{align}
	    The outer-connected dominating set for $G$ is $\{c\}$ and 	 the outer-connected dominating set for $G \Box H$ is $\{(c,x) ,(c,y)\}$. So,
	    	\begin{equation}
	    	\tilde{\gamma_c}(G \Box H) = \tilde{\gamma_c}(G) * |V(H)| = 1 * 2 = 2.
	    	\end{equation}		
	    		
	\end{example}
	
	\section{Outer Connected Domination in the Direct
		Product of Two Graphs}

	Gasper Mekis\cite{mekivs2010lower} gave a lower bound for the domination number of a direct product and proved the following result.
	\begin{theorem}\cite{mekivs2010lower}
		\label{tmek}
		Let $G = \times_{i=1}^t K_{n_i}$ , where $t \ge 3$ and $n_i \ge 2$ for all $i$. Then
		$\gamma(G) \ge t + 1.$
	\end{theorem}
	
	From Theorem \ref{tmek}  and that $\tilde{\gamma_c}(G) \ge \gamma(G)$, we obtain the following result.
	\begin{corollary}
		\label{direct}
			Let $G = \times_{i=1}^t K_{n_i}$ , where $t \ge 3$ and $n_i \ge 2$ for all $i$. Then
			$\tilde{\gamma_c}(G) \ge t + 1.$
			
	\end{corollary}
	
	The bound given in Theorem \ref{tmek} is sharp\cite{mekivs2010lower} and it also remains sharp for the outer-connected domination number.

	\begin{theorem}
	 Let $G = \times_{i=1}^t K_{n_i}$, where $t \ge 3$ and $n_i \ge t + 1$ for all $i$. Then
	$\tilde{\gamma_c}(G) = t + 1.$
	\end{theorem}
	\begin{proof}
		
   Let $\tilde{D} = \{(0, 0, \dots, 0), (1, 1, \dots, 1), \dots , (t, t, \dots , t)\}$ be an outer-connected dominatin set for $G$.
	Suppose that $ u = (u_1,u_2,\dots,u_t) \in V(G) \setminus \tilde{D} $ and $u$ is not adjacent to any of the vertices from $\tilde{D}$, in which case $u$ must agree
	in at least one coordinate with every vertex from $\tilde{D}$. Hence each of the $t + 1$ elements from $\{0, 1, \dots, t\}$ must appear on
	some coordinate of $u$, which is not possible as $u$ has only $t$ coordinates available. 
	Now, it suffices to show that $ V(G) \setminus \tilde{D}$ induces a connected graph.

Suppose that $u = (u_1,u_2,\dots,u_t), v = (v_1,v_2,\dots,v_t)$ are two arbitrary vertices in $ V(G) \setminus \tilde{D}$, we claim there exists a vertex 
	$x = (x_1,x_2,\dots,x_t)$ such that the path from $u$ to $v$ in $ V(G) \setminus \tilde{D}$ passes through it. Let $p_i$ be $i^{th}$ coordinate of the vertex $p$. the coordinate $x_i$ have to be three properties: (1) $x_i \ne u_i$, (2)  $x_i \ne v_i$, and  $x_i \ne x_{i-1}$ or $x_i \ne x_{i+1}$. So, if there exist at least four vertices in  all $K_{n_i}$ the claim is true which completes the proof. This end is clear since $n_i \ge t + 1$ and $t \ge 3$. 	
\end{proof}


	\end {document}